\documentclass[11pt]{article}

\usepackage{amsmath,amssymb,fullpage,times,color,graphicx,subfigure}
\usepackage{algorithm, algorithmic}
\usepackage[toc,page]{appendix}
\usepackage{tikz}
\usepackage{enumitem}
\usetikzlibrary{external}
\usetikzlibrary{arrows}
\newtheorem{theorem}{Theorem}[section]

\newtheorem{definition}[theorem]{Definition}

\newcommand{\qed}{\tag*{$\blacksquare$}}
\newenvironment{proof}{\par\noindent{\bf Proof.}}{  }

\title{An Improved Scheduling Algorithm for Traveling Tournament Problem with Maximum Trip Length Two}
\author{  
	Diptendu Chatterjee \thanks{Indian Statistical Institute,Kolkata, India.
		{\tt diptenduchatterjee@ymail.com}} 
  \and
	Bimal Kumar Roy \thanks{Indian Statistical Institute,Kolkata, India.
		{\tt bimal@isical.ac.in}} 
}

\begin{document}
	
	\maketitle
	
	\begin{abstract}
		The Traveling Tournament Problem(TTP) is a Combinatorial Optimization Problem where we have to give a scheduling algorithm which minimizes the total distance traveled by all participating teams of a double round-robin tournament maintaining given constraints. Most of the instances of this problem with more than ten teams is still unsolved. By definition of the problem the number of teams participating has to be even. There are different variants of this problem depending on the constraints. In this problem we consider the case where number of teams is a multiple of four and a team can not play more than two consecutive home or away matches. Our scheduling algorithm gives better result than the existing best result for number of teams less or equal to 32.
	\end{abstract}

	\section{Introduction}
	
	Double Round-robin tournament is one of the most unbiased way of evaluating teams participating in a competition. In this kind of tournament each of the participating team plays with every other team twice, i.e. one game in its home and another game in the home of the other team. This nullifies the effect of home ground and support. So, in this kind of tournament each team is tested in all the venues and in all the conditions. If there are $n$ teams participating, then each team will play $2(n-1)$ games and total number of games played will be $n(n-1)$. After all the matches are played, the team with highest point wins the tournament. Traveling Tournament Problem is inspired by \textit{Major League Baseball}. The general form of constrained Traveling Tournament Problem, i.e. $TTP-k$ for some natural number $k$, given participating teams and all the mutual distances between their home grounds is defined as follows.
	
	\begin{definition}\label{def:TTP_k} \textbf{TTP-k} is scheduling of a double round-robin tournament where total travel distance by all the participating teams is minimized given the following constraints:
\begin{enumerate}[wide, labelwidth=!, labelindent=0pt]
\setlength\itemsep{-0.2em}
\item Each pair of participating team play exactly two matches with each other once in each of their home venues.
\item No pair of teams play consecutive matches with each other.
\item In an away tour a visiting team travels directly from the home of one opponent to home of the next opponent without returning to its own home.
\item The lengths of the home stands and away tours for any participating team is not more than $k$.
\end{enumerate} 
	\end{definition}
	
	For odd number of teams scheduling of a Traveling Tournament Problem is not possible as in a match day every team should participate. 
	
	Like its benefits, Traveling Tournament Problem has some drawbacks also. The main drawbacks are huge number of matches and scheduling complexity. We can not decrease the number of matches, but we can lower the complexity of the scheduling. But with imposed constraints on scheduling the complexity increases. For a small number of teams the scheduling is simpler and the complexity increases with number of teams and imposed constraints. TTP-$\infty$ and and TTP-3 has been proven to be NP-hard in \cite{bhattacharyya2016complexity} and \cite{thielen2011complexity} respectively. TTP-1 is impossible to schedule \cite{de1988some}. So, the only possible case where complete solution may be possible is TTP-2. The complexity of TTP-2 is still not settled. The existing best result on approximating TTP-2 is given by Xiao and Kou \cite{xiao2016improved}. They gave an approximation factor of $(1+\frac{2}{n}+\frac{2}{n-2})$ for TTP-2 with n divisible by 4, where $n$ is the number of participating teams. We work on a similar setup, where we schedule a TTP-2 on $n$ teams with n divisible by 4 and our schedule improves the result for $n \leq 32$.
	
	Formal definition of the problem, some useful definitions, notations and some well known results related to Traveling Tournament Problem are given here.

	\subsection{Problem Definition}
	
	\textbf{TTP-2}: Traveling Tournament Problem-2 is scheduling of a double round-robin tournament where total travel distance by all the participating teams is minimized maintaining the following constraints:
	\begin{enumerate}[wide, labelwidth=!, labelindent=0pt]\setlength\itemsep{-0.2em}
		\item[\underline{Constraint 1}:] Each pair of participating team play exactly two matches with each other once in each of their home venues.
		\item[\underline{Constraint 2}:]No pair of teams play consecutive matches with each other.
		\item[\underline{Constraint 3}:] In an away tour a visiting team travels directly from the home of one opponent to home of the next opponent without returning to its own home.
		\item[\underline{Constraint 4}:] The lengths of the home stands and away tours for any participating team is not more than $2$. 
	\end{enumerate}
	
	\subsection{Previous Work}
	
	Traveling Tournament Problem(TTP) is a special variant of the Traveling Salesman Problem. The Traveling Tournament Problem was first introduced by Easton, Nemhauser, and Trick~\cite{easton2001traveling}. In a TTP, when there is no constraint on home stands or away trip length, it becomes a problem of scheduling $n$ Traveling Salesman Problem synchronously. It has been shown that, TTP-k i.e. Traveling Tournament Problem with not more than k home stands or away matches is NP-Hard when $K=\infty$ \cite{bhattacharyya2016complexity} or $k=3$ \cite{thielen2011complexity}. Relationship of some variants of round-robin tournaments with the planar three-index assignment problem has been analyzed and complexity of scheduling a minimum cost round-robin tournament has been established using the same \cite{briskorn2010round}. They also showed the applicability of some techniques for planar three-index assignment problem to solve a sub-problem of scheduling a minimum cost round-robin tournament . A large amount of work has been done towards the approximation algorithms \cite{xiao2016improved,imahori20142,hoshino2013approximation,miyashiro2012approximation,westphal20145,yamaguchi2011improved}. A large amount of work on heuristic algorithms also has been done \cite{anagnostopoulos2006simulated,di2007composite,easton2002solving,goerigk2014solving,lim2006simulated}. Many offline and online set of benchmark data set can be found for TTP-3 \cite{easton2001traveling,trick1999challenge}. For many benchmark result on improvements and complete solutions, work of high performance computers for more than a week is required \cite{van2007population}. But with that also most of the instances of TTP-k on more 10 teams are not completely solvable \cite{trick1999challenge}. They worked on a basketball tournament with ten teams where the away trip for any team consists of one or two matches. It is also ben shown that TTP-1 is impossible to schedule \cite{de1988some}.  A survey on round-robin tournament scheduling has been done by Rasmessen and Trick \cite{rasmussen2008round}. Work has also been done on complexity of TTP-k \cite{fujiwara2006constructive,hoshino2012linear,kendall2010scheduling}.
	
	Our main focus is on TTP-2 which was first introduced by Campbell and Chen \cite{campbell1976minimum}.Thielen and Westphal \cite{thielen2012approximation} has contributed towards approximation factor for TTP-2 and later gave an approximation factor of $(1+\tfrac{16}{n}); \forall n \geq 12$ and n divisible by 4. Their result has been improved by Xiao and Kou \cite{xiao2016improved}. They gave an approximation factor of  $(1+ \frac{2}{n-2}+ \frac{2}{n})$ where n is divisible by 4. Our scheduling algorithm give better result than this for $n \leq 32$.
	
	\subsection{Our Result}
	
	We propose a scheduling algorithm for TTP-2 which yields an approximation factor of $\left(1+ \frac{\left\lceil\log_{2}\frac{n}{4}\right\rceil+4}{2(n-2)}\right)$. For number of participating teams less or equal to $32$, this gives a better result than existing best result, with approximation factor of $(1+ \frac{2}{n-2}+ \frac{2}{n})$ in \cite{xiao2016improved}.
	
	\section{Preliminaries}
	
	\subsection{Definitions and Notations}
	
	In this paper, for getting better approximation factor for TTP-2, graph theoretic approach has been followed. Due to this, teams are invariably referred as vertices and distances between home locations of teams are referred as weights of edges of the graph. 
	
	\begin{definition}\label{def:matching}{\textbf{Matching Graph:}}
		A matching graph $G(V,E)$ is a graph where no two edges have a common vertex. So, for a matching graph, $|V|=n \Rightarrow |E| \leq \tfrac{n}{2}$. The pair of vertices connected through an edge in a matching graph is called matched vertices of the matching graph.	
	\end{definition}
	
	\begin{definition}\label{def:max_matching}{\textbf{Maximal Matching of a Graph:}}
		Maximal matching of a graph $G(V,E)$ is a matching of $G$, which is not subset of any other matching of $G$. It may not be unique for a given graph.	
	\end{definition}
	
	\begin{definition}\label{def:min_max_matching}{\textbf{Minimum Maximal Matching of an Undirected Weighted Graph:}}
		Minimum Maximal Matching of an Undirected Weighted Graph $G(V,E)$ is a maximal matching of $G$ with sum of all the weights of its edges is the smallest among that of all the maximal matching subgraphs of $G$. For a minimum maximal matching of an undirected weighted complete graph with $n$ vertices, the number of edges of the matching will be $\tfrac{n}{2}$.
	\end{definition}
	
	In this work, an edge between two vertices is represented as a match between the teams corresponding to the vertices. Now a \textit{super-match} is defined as follows:
	
	\begin{definition}\label{def:super_match}{\textbf{Super-match:}}
		A super-match between two pairs of matched vertices $M_i$ and $M_j$ is a set of edges $\{(u,w), (u,x), (v,w), (v,x)\}$ where $M_i=\{u,v\}$ and $M_j=\{w,x\}$.
	\end{definition}
	
	\subsection{A Simple Lower Bound for TTP-2}
	
	Let, there are $n$ teams participating in TTP-2. Distances between the home locations of each pair of teams are given. Let, $d_{ij}$ be the distance between home locations of $i^{th}$ and $j^{th}$ team. Now we construct an undirected weighted complete graph with all the $n$ home locations as vertices with weights of the edges as the physical distances between the home locations of teams corresponding to the vertices connected through it and call it $G(V,E)$. As $G$ is a complete graph and $|V|=n=$ even, we get a minimum maximal matching in $G$ and call it $G_m$. Let, sum of the weights of all the edges of $G_m$ be $W_m$, sum of the weights of all the edges in $G$ be $W_t$ and sum of the weights of all the edges from a vertex $i$ in $G$ be $W_i$.
	
	So,for an optimized schedule with the given constraints it is natural for a team to travel to two matched teams in $G_m$ in an away trip. But for the vertex matched with itself in $G_m$, it will make a to and fro journey. In that case, the total travel by $i^{th}$ team is $W_i+W_m$. This gives a minimum travel by $i$th team given the constraints.
	
	Now, if it is possible to synchronously fit this above mentioned minimum travel by each participating team in the schedule then the total traveled distance by all the teams in the tournament will be,
	\begin{equation*}
		\sum_{i \in V} (W_i+W_m)=2W_t+nW_m
	\end{equation*}
	
	This gives a lower bound to TTP-2. But due to the imposed constraints on scheduling and number of teams, it is not always possible to synchronously fit the minimum travel schedule of each participating teams in the schedule and here comes the optimization and hardness of the problem and makes this problem interesting.  
	
	\section{Design of Schedule}
	
	Suppose there are $n$ teams participating in a Double Round-robin Tournament where n is divisible by 4. We construct the undirected weighted graph $G(V,E)$ as described in \textbf{Section} $2$ and also find the minimum maximal matching $G_m$ in $G$. Now we number the vertices and the matched pairs such that matched pair $M_i$ consist of vertices $2i-1$ and $2i,\; \forall i \in \{1,\ldots, \tfrac{n}{2}\}$. Now, we design the schedule in $\lceil\log_{2}\tfrac{n}{2}\rceil$ rounds and $(\tfrac{n}{2}-1)$ levels such that $i^{th}$ round is consist of $\lceil \tfrac{1}{2} (\tfrac{n}{2^i}-1)\rceil$ levels and each level consists of $\tfrac{n}{4}$ \textit{super-matches}. A \textit{super-match} is played between two different matched pairs where both the teams in a matched pair plays home and away matches with both the teams in the other matched pair. In every level each matched pair plays a \textit{super-match}. We have designed three types of \textit{super-matches} which are used in our schedule. Suppose, there are two pairs of matched vertices $A_1,A_2$ and $B_1, B_2$ in $G_m$ described in the previous section. We give three types of \textit{super-match} namely Type-1, Type-2, Type-3 which are the building blocks of our schedule.
	
	\textbf{Type-1:} This consists of four match days namely $T_1, T_2, T_3$ and $T_4$ and the matches on this match days are given below:
	\begin{eqnarray}\nonumber
	T_1&:& A_1 \rightarrow B_1, A_2 \rightarrow B_2.\\ \nonumber
	T_2&:& A_1 \rightarrow B_2, A_2 \rightarrow B_1.\\ \nonumber
	T_3&:& B_1 \rightarrow A_1, B_2 \rightarrow A_2.\\ \nonumber
	T_4&:& B_1 \rightarrow A_2, B_2 \rightarrow A_1.
	\end{eqnarray}
	
	where $u \rightarrow v$ means $u$ is playing an away match with $v$ in the home of $v$.
	\\The home-away match sequence of the participating teams become the following:
	\\$A_1:aahh: A_2:aahh: B_1:hhaa: B_2:hhaa$. where $a$ means away match and $h$ means home match.
	\\Type-1 \textit{super-match} does not violate minimum travel of any of its participating teams.
	\\This way we can simultaneously schedule $\tfrac{n}{4}$ Type-1 \textit{super-matches} in a level but then we can not schedule matches between the teams with same home away match sequences due to \textbf{constraint:4} of the problem definition. So we need a different kind of \textit{super-match} like Type-1 and hence comes the need of Type-2 \textit{super-match}.
	
	\textbf{Type-2:} This consists of four match days namely $T_1, T_2, T_3$ and $T_4$ and the matches on this match days are given below:
	\begin{eqnarray}\nonumber
		T_1&:& A_1 \rightarrow B_1, A_2 \rightarrow B_2.\\ \nonumber
	    T_2&:& B_2 \rightarrow A_1, B_1 \rightarrow A_2.\\ \nonumber
	    T_3&:& B_1 \rightarrow A_1, B_2 \rightarrow A_2.\\ \nonumber
        T_4&:& A_1 \rightarrow B_2, A_2 \rightarrow B_1.
   \end{eqnarray} 
	The home-away match sequence of the participating teams become the following:
	\\$A_1:ahha: A_2:ahha: B_1:haah: B_2:haah$. where $a$ means away match and $h$ means home match.
	\\Type-2 \textit{super-match} violates minimum travel of all of its participating teams but helps to schedule matches of all the teams according to their minimum travel schedule in the next level. We may refer the Type-2 \textit{super-match} as \textit{flip} in future. But after this modification also it is not possible to schedule home and away matches between two matched teams in $G_m$ maintaining their minimum travel schedule. So there comes the need of Type-3 schedule block.
	
	\textbf{Type-3:} This consists of six match days namely $T_1, T_2, T_3, T_4,T_5$ and $T_6$ and the matches on this match days are given below:
	\begin{eqnarray}\nonumber
		T_1&:& A_1 \rightarrow B_1, A_2 \rightarrow B_2.\\ \nonumber
		T_2&:& A_1 \rightarrow A_2, B_2 \rightarrow B_1.\\ \nonumber
		T_3&:& B_2 \rightarrow A_1, B_1 \rightarrow A_2.\\ \nonumber
		T_4&:& A_2 \rightarrow A_1, B_1 \rightarrow B_2.\\ \nonumber
		T_5&:& A_1 \rightarrow B_2, A_2 \rightarrow B_1.\\ \nonumber
		T_6&:& B_1 \rightarrow A_1, B_2 \rightarrow A_2.
	\end{eqnarray} 
    The home-away match sequence of the participating teams become the following:
	\\$A_1:aahhah: A_2:ahhaah: B_1:hhaaha: B_2:haahha$.
	\\where $a$ means away match and $h$ means home match.  
	
	Although Type-1 \textit{super-match} does not violate the minimum travel schedule for the teams, we can not schedule a double round robin tournament only with Type-1 \textit{super-matches}. We need Type-2 and Type-3 \textit{super-matches}. Now, $\tfrac{n}{4}$ number of Type-3 \textit{super-matches} are unavoidable for any TTP-2 scheduling as each Type-3 \textit{super-match} involves home and away matches between matched vertices for two pairs of matched vertices of $G_m$ described in the previous section. So, for $n$ participating teams at least $\tfrac{n}{4}$ number of Type-3 \textit{super-matches} are required and our algorithm uses exactly $\tfrac{n}{4}$ numbers of Type-3 schedule blocks. Now, the only scope of improvement is reduction in numbers of Type-2 \textit{super-matches}. So our main aim to keep the the number of Type-2 \textit{super-matches} or \textit{flips} as low as possible.
	
	\section{Our Algorithm}
	Following algorithm gives a improved schedule in terms of total distance traveled by all the teams than the existing best result \cite{xiao2016improved} for TTP-2 when, $n \leq 32$ where the number of \textbf{Type-2} super matches are bounded by $\left(\frac{n}{8}*\left\lceil\log_{2}\frac{n}{4}\right\rceil\right)$ for all $n\in \mathbb{N}$. In next section, few schedules are given as examples using our algorithm.

	\begin{algorithm}[h]
		\caption{Schedule TTP-2}
		\label{alg:TTP-2}
		\begin{algorithmic}[1]
			\STATE \textbf{INPUT}: $G(V,E) \ \mbox{with} \ |V|=n, |E|={n \choose 2}, W=\{w_e|e \in E\}$.
			
			\STATE Identify the minimum maximal matching, $G_m(V,E_m)$, of $G$.
			
			\STATE $\forall i \in \{1,\dots,\frac{n}{2}\}, \mbox{define} \ M_i=\{(u,v)|u,v \in V \ \& \ \mbox{Edge}(u,v) \in E_m\}$.
			
			\STATE $\forall v \in V,\ \mbox{allot a number to} \ v \ \mbox{such that} \ (u,v) \in M_i \implies \#u=(2i-1) \ \& \ \#v=2i \  \forall i \in \{1,\dots,\frac{n}{2}\}$.
			
			\STATE Define $X=\{x_i|\mbox{location of} \ x_i \ \mbox{is in the midpoint of} \ u \ \& \ v \ \mbox{where} \ (u,v) \in M_i \ \forall i \in \{1,\dots,\frac{n}{2}\}\}$.
			
			\STATE Define a complete graph $H(X,E')|\forall e\in E'$, weight of the edge $e, W_e=$dist$(x_m,x_n)$ where $e$ is the edge between $x_m \ \& \ x_n$.
			
			\STATE Identify the minimum maximal matching, $H_m(X,E'_m)$, of $H$.
			
			\STATE $\forall i \in \{1,\dots,\frac{n}{4}\}, \mbox{define} N_i=\{(M_m,M_n)|x_m,x_n \in X \ \& \ \mbox{Edge}(x_m,x_n) \in E'_m\}$.
			
			\FOR {$i=1:1:\lceil\log_{2}\frac{n}{2}\rceil$}
			\WHILE {$2^{i+1} < n$}
			\IF {$2^{i+2}|n$}
			\STATE Schedule first $\left\lceil \frac{1}{2} \times (\frac{n}{2^i}-1)\right\rceil - 1$ levels of $i^{th}$ round each with $\frac{n}{4}$ Type-1 \textit{super-matches} and last level with $\frac{n}{8}$ Type-1 and $\frac{n}{8}$ Type-2 \textit{super-matches}. 
			\ELSE
			\STATE Schedule the $\left\lfloor \frac{n}{2}\sum_{1}^{i}{2^{-k}}-1 \right\rfloor ^{th}$ match days with $\left\lfloor\frac{n}{8}\right\rfloor$  $\mbox{Type-2 \textit{super-matches}}$ for $i\in\{1,2,\dots,\log_{2}n\}$ and rest of the super-matches as $\mbox{Type-1}$. For all other match days except the last one schedule all super-matches as $\mbox{Type-1}$.
			\ENDIF
			\STATE Schedule this last level of the tournament with $\frac{n}{4}$ Type-3 \textit{super-matches} where $\forall i \in \{1,\dots,\frac{n}{4}\}, M_p \ \mbox{plays with} \ M_q|M_p,M_q \in N_i$.
			\ENDWHILE
			\ENDFOR
			
		\end{algorithmic}
		
	\end{algorithm}
	
	In the above pseudo code for TTP-2 of n teams using our technique, first we find the Minimum Maximal Matching in the complete graph on all the vertices or teams. Let the set of matched pair of vertices be $\{M_1, \dots, M_{n/2}\}$. Then we consider each $M_i$'s as a team situated at the mid point of the locations of its constituent vertices. Then for a complete graph on these $M_i$'s as vertices, we again find the minimum maximal matching and let the set of matched vertices be $\{N_1, \dots, N_{n/4}\}$.
	
	Now, we schedule the Type-2 super-matches in the different levels of different rounds according to the rule described in line 12 or line 14 of the algorithm depending on the value of n. We schedule all the Type-3 super-matches in the last level of the last round of the tournament between matched pairs of $M_i$'s, i.e. between the elements of $N_i$'s to minimize the total travel distance.
	
	\section{Examples of Scheduling with Our Algorithm}
	
	For better understanding of our scheduling algorithm we give two examples of schedule for $n=12, 16$ here and $n=20,24,28$ in the Appendix-\ref{Examples_schedule}. An improved schedule of \textit{Indian Premier League}, where $n$=$8$, is presented in Appendix-\ref{IPL Schedule}. Let,
	\begin{equation}
		\label{eq:Fn}
		F_n=\frac{n}{8}*\left\lceil\log_{2}\frac{n}{4}\right\rceil \quad \mbox{for } n\in \mathbb{N}.
	\end{equation}

	\subsection{Schedule for $n=12$}
	For designing a Traveling Tournament Problem of 12 teams using our technique, first we number the teams or the vertices with natural numbers as follows.
	
	Vertex Set=\{1, 2, 3, 4, 5, 6, 7, 8, 9, 10, 11, 12\}. 
	
	Then we find the Minimum Maximal Matching in the complete graph containing the vertices in the above mentioned vertex set. Let the set of matched pair of vertices be $\{M_1, M_2, M_3, M_4, M_5, M_6\}$ and without loss of generality we can say that
	
	$M_1$=\{1,2\}, $M_2$=\{3,4\}, $M_3$=\{5,6\}, $M_4$=\{7,8\}, $M_5$=\{9,10\}, $M_6$=\{11,12\}
	
	Then we consider each $M_i$'s as a team situated at the mid point of the locations of its constituent vertices for $i\in\{1,2,3,4,5,6\}$. Then for a complete graph on these $M_i$'s as vertices, we find the minimum maximal matching and let the set of matched vertices be 
	
	$\{N_1, N_2,N_3\}$ such that $N_1$=$\{M_1, M_5\}$, $N_2$=$\{M_2, M_3\}$, $N_3$=$\{M_4, M_6\}$.
	
	Now, we describe below the super-matches to be scheduled in all the levels of all the rounds according to our scheduling technique in a tabular form. We can observe that the super-matches scheduled in the last level of the last round of the tournament are between matched pairs of $M_i$'s, i.e. between the elements of $N_i$'s.
	
	\begin{table}[h]
		\centering
		\begin{tabular}{ccc}
			\textbf{Round:1, Level:1}                                                                                                                & \textbf{Round:1, Level:2}                                                                                                                & \textbf{Round:1, Level:3}                                                                                                                \\
			\begin{tabular}[c]{@{}l@{}}$M_1 \xrightarrow{Type-1} M_2$\\ $M_3 \xrightarrow{Type-1} M_4$\\ $M_5 \xrightarrow{Type-1} M_6$\end{tabular} & \begin{tabular}[c]{@{}l@{}}$M_1 \xrightarrow{Type-1} M_4$\\ $M_3 \xrightarrow{Type-2} M_6$\\ $M_5 \xrightarrow{Type-1} M_2$\end{tabular} & \begin{tabular}[c]{@{}l@{}}$M_1 \xrightarrow{Type-1} M_3$\\ $M_6 \xrightarrow{Type-2} M_2$\\ $M_5 \xrightarrow{Type-1} M_4$\end{tabular} \\
			&                                                                                                                                          &                                                                                                                                          \\
			\textbf{Round:2, Level:1}                                                                                                                &\textbf{Round:3, Level:1}                                                                                                                & \textbf{}                                                                                                                                \\
			\begin{tabular}[c]{@{}l@{}}$M_1 \xrightarrow{Type-2} M_6$\\ $M_2 \xrightarrow{Type-1} M_4$\\ $M_5 \xrightarrow{Type-1} M_3$\end{tabular}&
			\begin{tabular}[c]{@{}l@{}}$M_6 \xrightarrow{Type-3} M_4$\\ $M_2 \xrightarrow{Type-3} M_3$\\ $M_5 \xrightarrow{Type-3} M_1$\end{tabular} &  \\
			& & \\
			\multicolumn{3}{l}{	\textbf{Number of \textit{Flips}$=3=F_{12}$.}}                                	                                                                                                                                       
		\end{tabular}
	\end{table}
	
	\subsection{Schedule for $n=16$}
	Now for designing a Traveling Tournament Problem of 16 teams using our technique, first we number the teams or the vertices with natural numbers in a similar fashion as follows.
	
	Vertex Set=\{1,2,3,4,5,6,7,8,9,10,11,12,13,14,15,16\}. 
	
	Then we find the Minimum Maximal Matching in the complete graph containing the vertices in the above mentioned vertex set. Let the set of matched pair of vertices be $\{M_1, M_2, M_3, M_4, M_5, M_6, M_7, M_8\}$ and without loss of generality we can say that
	
	$M_1$=\{1,2\}, $M_2$=\{3,4\}, $M_3$=\{5,6\}, $M_4$=\{7,8\}, $M_5$=\{9,10\}, $M_6$=\{11,12\}, $M_7$=\{13,14\}, $M_8$=\{15,16\}
	
	Then we consider each $M_i$'s as a team situated at the mid point of the locations of its constituent vertices for $i\in\{1,2,3,4,5,6,7,8\}$. Then for a complete graph on these $M_i$'s as vertices, we find the minimum maximal matching and let the set of matched vertices be 
	
	$N_1$=$\{M_1, M_5\}$, $N_2$=$\{M_2, M_6\}$, $N_3$=$\{M_3, M_7\}$, $N_4$=$\{M_4, M_8\}$.
	
	Now, we describe below the super-matches to be scheduled in all the levels of all the rounds according to our scheduling technique in a tabular form. We can observe that the super-matches scheduled in the last level of the last round of the tournament are between matched pairs of $M_i$'s, i.e. between the elements of $N_i$'s. Also as 8 is a power of 2, we exactly know the super-matches which are \textit{flips} in the different levels of all the rounds of the tournament according to our scheduling technique.
	\begin{table}[h]
		\centering
		\begin{tabular}{cccc}
			\textbf{Round:1, Level:1}                                                                                                                                                 & \textbf{Round:1, Level:2}                                                                                                                                                 & \textbf{Round:1, Level:3}                                                                                                                                                 & \textbf{Round:1, Level:4}                                                                                                                                                 \\
			\begin{tabular}[c]{@{}c@{}}$M_1 \xrightarrow{Type-1} M_2$\\ $M_3 \xrightarrow{Type-1} M_4$\\ $M_5 \xrightarrow{Type-1} M_6$\\ $M_7 \xrightarrow{Type-1} M_8$\end{tabular} & \begin{tabular}[c]{@{}c@{}}$M_1 \xrightarrow{Type-1} M_4$\\ $M_3 \xrightarrow{Type-1} M_6$\\ $M_5 \xrightarrow{Type-1} M_8$\\ $M_7 \xrightarrow{Type-1} M_2$\end{tabular} & \begin{tabular}[c]{@{}c@{}}$M_1 \xrightarrow{Type-1} M_6$\\ $M_3 \xrightarrow{Type-1} M_8$\\ $M_5 \xrightarrow{Type-1} M_2$\\ $M_7 \xrightarrow{Type-1} M_4$\end{tabular} & \begin{tabular}[c]{@{}c@{}}$M_1 \xrightarrow{Type-1} M_8$\\ $M_3 \xrightarrow{Type-2} M_2$\\ $M_5 \xrightarrow{Type-1} M_4$\\ $M_7 \xrightarrow{Type-2} M_6$\end{tabular} \\& & &                                                                                               \\
			\textbf{Round:2, Level:1}                                                                                                                                                 & \textbf{Round:2, Level:2}                                                                                                                                                 & \textbf{Round:3, Level:1}                                                                                                                                                 &  \\
			\begin{tabular}[c]{@{}c@{}}$M_1 \xrightarrow{Type-1} M_3$\\ $M_5 \xrightarrow{Type-1} M_7$\\ $M_2 \xrightarrow{Type-1} M_8$\\ $M_6 \xrightarrow{Type-1} M_4$\end{tabular} & \begin{tabular}[c]{@{}c@{}}$M_1 \xrightarrow{Type-1} M_7$\\ $M_5 \xrightarrow{Type-2} M_3$\\ $M_2 \xrightarrow{Type-1} M_4$\\ $M_6 \xrightarrow{Type-2} M_8$\end{tabular} & \begin{tabular}[c]{@{}c@{}}$M_1 \xrightarrow{Type-3} M_5$\\ $M_3 \xrightarrow{Type-3} M_7$\\ $M_2 \xrightarrow{Type-3} M_6$\\ $M_8 \xrightarrow{Type-3} M_4$\end{tabular} & \\
			& & &\\
			\multicolumn{4}{l}{	\textbf{Number of \textit{Flips}$=4=F_{16}$.}}                                	                                                                                                                                                                                     
		\end{tabular}
	\end{table}
	
Correctness of this algorithm is assured by the structures of Type-1, Type-2 and Type-3 \textit{super-matches}. As all three of these structures do not violate any of the constraints in the problem definition, so our schedule also does not violate any of the constraints. Which proves the correctness of our algorithm. 
	\section{Proof of Results}
	Theorems related to the analysis of the proposed algorithm along with their proofs are presented in this section.
	\begin{theorem}\label{thm:type_3_bound}
		All the Type-3 schedule blocks together introduce a relative error at most $\tfrac{2}{n-2}$  times of the Lower Bound of TTP-2.
	\end{theorem}
	\begin{proof}
		Suppose for some $i \in \{1, \ldots, \tfrac{n}{4}\}, N_i$ includes $4$ vertices of $G$ i.e. $A_1, A_2, B_1, B_2$ where $A_1$ and $A_2$ are matched pairs in $G_m$ and so are $B_1$ and $B_2$. For a Type-3 schedule in between them, travel for each team are given below:
		\begin{eqnarray}\nonumber
			A_1&:& A_1 \rightarrow B_1 \rightarrow A_2 \rightarrow A_1 \rightarrow B_2 \rightarrow A_1.\\ \nonumber
			A_2&:& A_2 \rightarrow B_2 \rightarrow A_2 \rightarrow A_1 \rightarrow B_1 \rightarrow A_2.\\ \nonumber
			B_1&:& B_1 \rightarrow A_2 \rightarrow B_2 \rightarrow B_1 \rightarrow A_1 \rightarrow B_1.\\ \nonumber
			B_2&:& B_2 \rightarrow B_1 \rightarrow A_1 \rightarrow B_2 \rightarrow A_2 \rightarrow B_2.
		\end{eqnarray}
		So the total distance traveled is,
		\begin{equation*}
			5*dist(A_1, B_1)+ 3*dist(A_2, B_1)+ 2*dist(A_1, A_2)+ 3*dist(A_1, B_2)+5*dist(A_2, B_2)+2*dist(B_1, B_2)
		\end{equation*}
		For the minimum travel schedule the value is,
		\begin{equation*}
			2*dist(A_1, B_1)+ 2*dist(A_2, B_1)+ 6*dist(A_1, A_2)+ 2*dist(A_1, B_2)+2*dist(A_2, B_2)+6*dist(B_1, B_2)
		\end{equation*}
		So the extra amount of travel is,
		\begin{equation*}
			3*dist(A_1, B_1)+ 1*dist(A_2, B_1)- 4*dist(A_1, A_2)+ 1*dist(A_1, B_2)+3*dist(A_2, B_2)- 4*dist(B_1, B_2)
		\end{equation*}
		Using triangle inequality,the above expression is upper bounded by,
		\begin{equation*}
			2*dist(A_1, B_1)+ 2*dist(A_2, B_2)+ 2*dist(A_1, B_2)+ 2*dist(A_2, B_1)
		\end{equation*}
		Let us denote, \textit{super-edge} $D_{ij}$ between pairs $A_1, A_2$ and $B_1, B_2$ as,
		\begin{equation*}
			dist(A_1, B_1)+ dist(A_2, B_2)+ dist(A_1, B_2)+ dist(A_2, B_1)
		\end{equation*}
		where
		\begin{equation*}
			A_1, A_2 \in M_i\ \mbox{and}\ B_1, B_2 \in M_j\ \mbox{for some}\ i,j \in \{1,\ldots,\tfrac{n}{2}\}
		\end{equation*}
		
		Now, there are $\tfrac{n}{2}$ numbers of pair of vertices like $A_1, A_2$. If we consider all pairwise distances between all these $\tfrac{n}{2}$ pairs, then we get all the edges of the complete graph $G$ but the edges of the matching $G_m$. But among all these $\binom{n/2}{2}$ pairwise distances, we are interested in $\tfrac{n}{4}$ matched pairwise distances as described in line $16$ of algorithm \ref{alg:TTP-2}, while calculating the error due to all Type-3 schedule blocks. So, the total error due to Type-3 schedule blocks is bounded by
		\begin{equation}\nonumber
			2*\frac{n/4}{\binom{n/2}{2}}*(W_t-W_m) < \frac{2}{n-2}*\mbox{(Lower Bound of TTP-2)}.
			\qed
		\end{equation} 
	\end{proof}
	\begin{theorem}\label{thm:type_2_and_type_3_bound}
		All the Type-2 and Type-3 schedule blocks together introduces relative error at most $\frac{\left\lceil\log_{2}\frac{n}{4}\right\rceil+4}{2(n-2)}$  times of the Lower Bound of TTP-2.
	\end{theorem}
	\begin{proof}
		Suppose a Type-2 schedule block is designed among $4$ vertices of $G$ i.e. $A_1, A_2, B_1, B_2$ where $A_1$ and $A_2$ are matched pairs in $G_m$ and so are $B_1$ and $B_2$. For a Type-2 schedule in between them, travel for each team are given below:
		\begin{eqnarray}\nonumber
			A_1&:& A_1 \rightarrow B_1 \rightarrow A_1 \rightarrow B_2 \rightarrow A_1.\\ \nonumber
			A_2&:& A_2 \rightarrow B_2 \rightarrow A_2 \rightarrow B_1 \rightarrow A_2.\\ \nonumber
			B_1&:& B_1 \rightarrow A_2 \rightarrow A_1 \rightarrow B_1.\\ \nonumber
			B_2&:& B_2 \rightarrow A_1 \rightarrow A_2 \rightarrow B_2.
		\end{eqnarray}	
	
		So the total distance traveled is,
		\begin{equation*}
			3*dist(A_1, B_1)+ 3*dist(A_2, B_1)+ 2*dist(A_1, A_2)+ 3*dist(A_1, B_2)+3*dist(A_2, B_2)
		\end{equation*}
		For the minimum travel schedule the value is,
		\begin{equation*}
			2*dist(A_1, B_1)+ 2*dist(A_2, B_1)+ 2*dist(A_1, A_2)+ 2*dist(A_1, B_2)+2*dist(A_2, B_2)+2*dist(B_1, B_2)
		\end{equation*}
		So the extra amount of travel is,
		\begin{equation*}
			dist(A_1, B_1)+ dist(A_2, B_1)+ dist(A_1, B_2)+ dist(A_2, B_2)- 2*dist(B_1, B_2)
		\end{equation*}
		Which is upper bounded by,
		\begin{equation*}
			dist(A_1, B_1)+ dist(A_2, B_2)+ dist(A_1, B_2)+ dist(A_2, B_1)
		\end{equation*}
		Let us denote the pairwise distance $D_P(A,B)$ between pairs $A_1, A_2$ and $B_1, B_2$ as, 
		\begin{equation*}
			dist(A_1, B_1)+ dist(A_2, B_2)+ dist(A_1, B_2)+ dist(A_2, B_1)
		\end{equation*}
		Now, there are $\frac{n}{2}$ numbers of pair of vertices like $A_1, A_2$. If we consider all pairwise distances between all these $\tfrac{n}{2}$ pairs, then we get all the edges of the complete graph $G$ but the edges of the matching $G_m$. But among all these $\binom{n/2}{2}$ pairwise distances, we have already selected $\tfrac{n}{4}$ pairwise distances as described in the proof of Theorem \ref{thm:type_3_bound} and now we are interested in at most $F_n$, given in equation \ref{eq:Fn}, pairwise distances as per line $12$ or $14$ of algorithm \ref{alg:TTP-2}, while calculating the error due to all Type-2 schedule blocks. So, the total error due to Type-2 and Type-3 schedule blocks is bounded by,
		\begin{equation}\nonumber
			\frac{\frac{n}{8}*\lceil\log_{2}\frac{n}{4}\rceil+\frac{n}{2}}{\binom{n/2}{2}}*(W_t-W_m) < \frac{\lceil\log_{2}\frac{n}{4}\rceil+4}{2(n-2)}*\mbox{(Lower Bound of TTP-2)}. \qed
	\end{equation} \end{proof}  
	\begin{theorem}\label{thm:comparision_approx_factor}
		Our algorithm gives better approximation than existing best result for number of participating teams less than or equal to $32$.
	\end{theorem}
	
	\begin{proof}
		From the last two theorems we can see that the approximation factor in our algorithm is $1+ \tfrac{\left\lceil\log_{2}\frac{n}{4}\right\rceil+4}{2(n-2)}$ and in the existing best result the approximation factor is $1+ \tfrac{2}{n-2}+ \tfrac{2}{n}$ \cite{xiao2016improved}. So for $n\leq32$,
		\begin{eqnarray}\nonumber
			&&\frac{8}{n} \leq \left\lfloor\log_{2}\frac{64}{n}\right\rfloor\Longleftrightarrow \frac{8}{n} \leq 4-\left\lceil\log_{2}\frac{n}{4}\right\rceil\Longleftrightarrow \left\lceil\log_{2}\frac{n}{4}\right\rceil \leq 4-\frac{8}{n} \Longleftrightarrow \left\lceil\log_{2}\frac{n}{4}\right\rceil \leq \frac{4}{n}(n-2) \\ \nonumber &&\Longleftrightarrow \frac{\left\lceil\log_{2}\frac{n}{4}\right\rceil}{(n-2)} \leq \frac{4}{n} \Longleftrightarrow \frac{\left\lceil\log_{2}\frac{n}{4}\right\rceil}{2(n-2)} \leq \frac{2}{n}
			\Longleftrightarrow\frac{4+\left\lceil\log_{2}\frac{n}{4}\right\rceil}{2(n-2)} \leq \frac{2}{n-2}+ \frac{2}{n}  
		\end{eqnarray}
		This proves the theorem. \hfill $\blacksquare$
	\end{proof}
	
	\section{Conclusion}
	
	In this work, a better approximation factor than the existing best result has been achieved for Traveling Tournament Problem with maximum trip length two with our scheduling algorithm when the number of participating team is less or equal to 32. Due to time constraints and other factors, most of the tournaments involving number of teams more than $32$ are not Round-Robin tournaments. For example a round-robin tournament with $40$ teams will require $78$ match days, $1560$ matches and $40$ grounds which demand lots of time, human support and a very long season. That is why most of the Round-Robin tournaments are conducted with less than $32$ teams. Therefore, it can be said that for almost all practical cases the proposed scheduling algorithm would produce better result than the existing best result. One of the popular double round-robin tournament in India is \textbf{Indian Premier League(IPL)} and the number of teams involved in this tournament is $8$. This tournament is not in TTP-2 structure now. But, if it is scheduled in TTP-2 structure, the proposed algorithm will significantly lower the total travel distance. An improved schedule of \textbf{IPL} using the proposed scheduling algorithm is presented in Appendix-\ref{IPL Schedule}. It shows a $15\%$ decrease in total travel distance in comparison with the actual \textbf{IPL-2019} schedule.
	
	\section{Scope of Future Work}
	
	As described in our algorithm, we know the specific match days of the schedule where the \textbf{Type-2} super matches or \textit{Flips} are to be incorporated. But as we have specified the pairs of teams between whom the \textbf{Type-3} super matches are to be played to minimize the total travel distance due to the \textbf{Type-3} super matches, nothing of this kind is done for the \textit{Flips}. So, a revisit in this topic can give some idea about the specific pairs of teams for minimizing the distance due to the \textit{Flips}.
	\bibliographystyle{unsrt}
	\bibliography{ttp}
	
	\appendix
	
	\section{More Examples of Schedule}
	\label{Examples_schedule}
	Schedules for $n=20,24,28$ are given below for a better insight of our algorithm.
	
	\subsection{Schedule for $ n = 20 $}
	Vertex Set $=\{1,2,3,4,5,6,7,8,9,10,11,12,13,14,15,16,17,18,19,20\}$.\\ Set of Pair of vertices $=\{M_1, M_2, M_3, M_4, M_5, M_6, M_7, M_8, M_9, M_{10}\}$;\\ where $M_1=\{1,2\}, M_2=\{3,4\}, M_3=\{5,6\}, M_4=\{7,8\}, M_5=\{9,10\}, M_6=\{11,12\},$ $ M_7=\{13,14\}, M_8=\{15,16\}, M_9=\{17,18\}, M_{10}=\{19,20\}$;\\ and $N_1=\{M_1, M_5\}, N_2=\{M_2, M_{10}\}, N_3=\{M_3, M_9\}, N_4=\{M_4, M_7\}, N_5=\{M_6, M_8\}$.
	
		\begin{tabular}{ccc}
			\centering
			\textbf{Round:1, Level:1}                                                                                                                                                                                     & \textbf{Round:1, Level:2}                                                                                                                                                                                     & \textbf{Round:1, Level:3}                                                                                                                                                                                                                                                                                                                                                     \\
			\begin{tabular}[c]{@{}c@{}}$M_1 \xrightarrow{Type-1} M_2$\\ $M_3 \xrightarrow{Type-1} M_4$\\ $M_5 \xrightarrow{Type-1} M_6$\\ $M_7 \xrightarrow{Type-1} M_8$\\ $M_9 \xrightarrow{Type-1} M_{10}$\end{tabular} & \begin{tabular}[c]{@{}c@{}}$M_1 \xrightarrow{Type-1} M_4$\\ $M_3 \xrightarrow{Type-1} M_6$\\ $M_5 \xrightarrow{Type-1} M_8$\\ $M_7 \xrightarrow{Type-1} M_{10}$\\ $M_9 \xrightarrow{Type-1} M_2$\end{tabular} & \begin{tabular}[c]{@{}c@{}}$M_1 \xrightarrow{Type-1} M_6$\\ $M_3 \xrightarrow{Type-2} M_8$\\ $M_5 \xrightarrow{Type-1} M_{10}$\\ $M_7 \xrightarrow{Type-2} M_2$\\ $M_9 \xrightarrow{Type-1} M_4$\end{tabular} \\
			& &\\
			\textbf{Round:1, Level:4}                                                                                                                                                                                     & \textbf{Round:1, Level:5}
			&\textbf{Round:2, Level:1}
			\\
			\begin{tabular}[c]{@{}c@{}}$M_1 \xrightarrow{Type-1} M_3$\\ $M_8 \xrightarrow{Type-2} M_{10}$\\ $M_5 \xrightarrow{Type-1} M_7$\\ $M_2 \xrightarrow{Type-1} M_4$\\ $M_9 \xrightarrow{Type-1} M_6$\end{tabular} & \begin{tabular}[c]{@{}l@{}}$M_1 \xrightarrow{Type-1} M_8$\\ $M_{10} \xrightarrow{Type-1} M_4$\\ $M_5 \xrightarrow{Type-1} M_3$\\ $M_2 \xrightarrow{Type-1} M_6$\\ $M_9 \xrightarrow{Type-1} M_7$\end{tabular} &\begin{tabular}[c]{@{}c@{}}$M_1 \xrightarrow{Type-2} M_7$\\ $M_{10} \xrightarrow{Type-1} M_6$\\ $M_5 \xrightarrow{Type-2} M_4$\\ $M_2 \xrightarrow{Type-1} M_3$\\ $M_9 \xrightarrow{Type-1} M_8$\end{tabular} 
		\\ 	& &\\
			\textbf{Round:2, Level:2}                                                                                                                                                                                     & \textbf{Round:3, Level:1}                                                                                                                                                                                     & \textbf{Round:4, Level:1}                                                                                                                                                                                                                    \\
			\begin{tabular}[c]{@{}c@{}}$M_7 \xrightarrow{Type-1} M_3$\\ $M_{10} \xrightarrow{Type-1} M_1$\\ $M_4 \xrightarrow{Type-1} M_6$\\ $M_2 \xrightarrow{Type-1} M_8$\\ $M_9 \xrightarrow{Type-1} M_5$\end{tabular} & \begin{tabular}[c]{@{}c@{}}$M_7 \xrightarrow{Type-1} M_6$\\ $M_{10} \xrightarrow{Type-1} M_3$\\ $M_4 \xrightarrow{Type-2} M_8$\\ $M_2 \xrightarrow{Type-2} M_5$\\ $M_9 \xrightarrow{Type-1} M_1$\end{tabular} & \begin{tabular}[c]{@{}c@{}}$M_7 \xrightarrow{Type-3} M_4$\\ $M_{10} \xrightarrow{Type-3} M_2$\\ $M_8 \xrightarrow{Type-3} M_6$\\ $M_5 \xrightarrow{Type-3} M_1$\\ $M_9 \xrightarrow{Type-3} M_3$\end{tabular}\\
			& & \\
			\multicolumn{3}{l}{	\textbf{Number of \textit{Flips}$=7=\lfloor F_{20} \rfloor$.}}                                	
		\end{tabular}

	\subsection{Schedule for $n=24$}
	Vertex Set =\{1, 2, 3, 4, 5, 6, 7, 8, 9, 10, 11, 12, 13, 14, 15, 16, 17, 18, 19, 20, 21, 22, 23, 24\}.\\ Set of Pair of vertices $=\{M_1, M_2, M_3, M_4, M_5, M_6, M_7, M_8, M_9, M_{10}, M_{11}, M_{12}\}$;\\where $M_1$=\{1,2\}, $M_2$=\{3,4\}, $M_3$=\{5,6\}, $M_4$=\{7,8\}, $M_5$=\{9,10\}, $M_6$=\{11,12\}, $M_7$=\{13,14\}, $M_8$=\{15,16\}, $M_9$=\{17,18\}, $M_{10}$=\{19,20\}, $M_{11}$=\{21,22\}, $M_{12}$=\{23,24\}; 
	\\and $N_1$=$\{M_9, M_5\}$, $N_2$=$\{M_1, M_7\}$, $N_3$=$\{M_{11}, M_3\}$, $N_4$=$\{M_{10}, M_6\}$, $N_5$=$\{M_2, M_4\}$, $N_6$=$\{M_8, M_{12}\}$.	
	\begin{table}[htbp]
		\centering
		\begin{tabular}{cccc}
			\textbf{Round:1, Level:1}                                                                                                                                                                                                                            & \textbf{Round:1, Level:2}                                                                                                                                                                                                                            & \textbf{Round:1, Level:3}                                                                                                                                                                                                                            & \textbf{Round:1, Level:4}                                                                                                                                                                                                                                                                                                                                                                                                                                  \\
			\begin{tabular}[c]{@{}c@{}}$M_1 \xrightarrow{Type-1} M_2$\\ $M_3 \xrightarrow{Type-1} M_4$\\ $M_5 \xrightarrow{Type-1} M_6$\\ $M_7 \xrightarrow{Type-1} M_8$\\ $M_9 \xrightarrow{Type-1} M_{10}$\\ $M_{11} \xrightarrow{Type-1} M_{12}$\end{tabular} & \begin{tabular}[c]{@{}c@{}}$M_1 \xrightarrow{Type-1} M_4$\\ $M_3 \xrightarrow{Type-1} M_6$\\ $M_5 \xrightarrow{Type-1} M_8$\\ $M_7 \xrightarrow{Type-1} M_{10}$\\ $M_9 \xrightarrow{Type-1} M_{12}$\\ $M_{11} \xrightarrow{Type-1} M_2$\end{tabular} & \begin{tabular}[c]{@{}c@{}}$M_1 \xrightarrow{Type-1} M_6$\\ $M_3 \xrightarrow{Type-1} M_8$\\ $M_5 \xrightarrow{Type-1} M_{10}$\\ $M_7 \xrightarrow{Type-1} M_{12}$\\ $M_9 \xrightarrow{Type-1} M_2$\\ $M_{11} \xrightarrow{Type-1} M_4$\end{tabular} & \begin{tabular}[c]{@{}c@{}}$M_1 \xrightarrow{Type-1} M_8$\\ $M_3 \xrightarrow{Type-1} M_{10}$\\ $M_5 \xrightarrow{Type-1} M_{12}$\\ $M_7 \xrightarrow{Type-1} M_2$\\ $M_9 \xrightarrow{Type-1} M_4$\\ $M_{11} \xrightarrow{Type-1} M_6$\end{tabular}	                              
		\end{tabular}
	\end{table}

	\begin{table}[htbp]
		\centering
		\begin{tabular}{cccc}
			\textbf{Round:1, Level:5}                                                                                                                                                                                                        & \textbf{Round:1, Level:6}
			& 	\textbf{Round:2, Level:1}                                                                                                                                                                                                                            & \textbf{Round:2, Level:2}                                                                                                                                                                                                                             \\
			\begin{tabular}[c]{@{}l@{}}$M_1 \xrightarrow{Type-1} M_{10}$\\ $M_3 \xrightarrow{Type-1} M_{12}$\\ $M_5 \xrightarrow{Type-1} M_2$\\ $M_7 \xrightarrow{Type-1} M_4$\\ $M_9 \xrightarrow{Type-1} M_6$\\ $M_{11} \xrightarrow{Type-1} M_8$\end{tabular} & \begin{tabular}[c]{@{}l@{}}$M_1 \xrightarrow{Type-2} M_{12}$\\ $M_3 \xrightarrow{Type-1} M_2$\\ $M_5 \xrightarrow{Type-2} M_4$\\ $M_7 \xrightarrow{Type-1} M_6$\\ $M_9 \xrightarrow{Type-2} M_8$\\ $M_{11} \xrightarrow{Type-1} M_{10}$\end{tabular} &
			\begin{tabular}[c]{@{}c@{}}$M_{12} \xrightarrow{Type-1} M_2$\\ $M_3 \xrightarrow{Type-1} M_1$\\ $M_4 \xrightarrow{Type-1} M_6$\\ $M_7 \xrightarrow{Type-1} M_5$\\ $M_8 \xrightarrow{Type-1} M_{10}$\\ $M_{11} \xrightarrow{Type-1} M_9$\end{tabular} & \begin{tabular}[c]{@{}c@{}}$M_{12} \xrightarrow{Type-1} M_6$\\ $M_3 \xrightarrow{Type-1} M_5$\\ $M_4 \xrightarrow{Type-2} M_{10}$\\ $M_7 \xrightarrow{Type-2} M_9$\\ $M_8 \xrightarrow{Type-1} M_2$\\ $M_{11} \xrightarrow{Type-1} M_1$\end{tabular}	
	\\ & & &\\
			\textbf{Round:2, Level:3}                                                                                                                                                                                                                            & \textbf{Round:3, Level:1}                                                                                                                                                                                                                            & \multicolumn{1}{c}{\textbf{Round:4, Level:1}}                                                                                                                                                                                                        &                                                          \\
			\begin{tabular}[c]{@{}c@{}}$M_{12} \xrightarrow{Type-1} M_4$\\ $M_3 \xrightarrow{Type-1} M_7$\\ $M_{10} \xrightarrow{Type-2} M_2$\\ $M_9 \xrightarrow{Type-2} M_1$\\ $M_8 \xrightarrow{Type-1} M_6$\\ $M_{11} \xrightarrow{Type-1} M_5$\end{tabular} & \begin{tabular}[c]{@{}c@{}}$M_{12} \xrightarrow{Type-2} M_{10}$\\ $M_3 \xrightarrow{Type-2} M_9$\\ $M_2 \xrightarrow{Type-1} M_6$\\ $M_1 \xrightarrow{Type-1} M_5$\\ $M_8 \xrightarrow{Type-1} M_4$\\ $M_{11} \xrightarrow{Type-1} M_7$\end{tabular} & \begin{tabular}[c]{@{}l@{}}$M_9 \xrightarrow{Type-3} M_5$\\ $M_1 \xrightarrow{Type-3} M_7$\\ $M_{11} \xrightarrow{Type-3} M_3$\\ $M_{10} \xrightarrow{Type-3} M_6$\\ $M_2 \xrightarrow{Type-3} M_4$\\ $M_8 \xrightarrow{Type-3} M_{12}$\end{tabular} &  \\
			& & & \\
			\multicolumn{4}{l}{	\textbf{Number of \textit{Flips}$=9=F_{24}$.}}                                     
		\end{tabular}	
	\end{table}
	    
	\subsection{Schedule for $n=28$}
	Vertex Set =\{1, 2, 3, 4, 5, 6, 7, 8, 9, 10, 11, 12, 13, 14, 15, 16, 17, 18, 19, 20, 21, 22, 23, 24, 25, 26, 27, 28\}.\\ Set of Pair of vertices $=\{M_1, M_2, M_3, M_4, M_5, M_6, M_7, M_8, M_9, M_{10}, M_{11}, M_{12}, M_{13}, M_{14}\}$;\\where $M_1$=\{1,2\}, $M_2$=\{3,4\}, $M_3$=\{5,6\}, $M_4$=\{7,8\}, $M_5$=\{9,10\}, $M_6$=\{11,12\}, $M_7$=\{{13},14\}, $M_8$=\{15,16\}, $M_9$=\{17,18\}, $M_{10}$=\{19,20\}, $M_{11}$=\{21,22\}, $M_{12}$=\{23,24\}, $M_{13}$=\{25,26\}, \\$M_{14}$=\{27,28\};\\and $N_1$=$\{M_{14}, M_4\}$, $N_2$=$\{M_{12}, M_6\}$, $N_3$=$\{M_3, M_2\}$, $N_4$=$\{M_8, M_{10}\}$, $N_5$=$\{M_9, M_5\}$, \\$N_6$=$\{M_7, M_{11}\}$, $N_7$=$\{M_1, M_{13}\}$.
		
	\begin{table}[htbp]
		\centering
		\begin{tabular}{cccc}
			\textbf{Round:1, Level:1}                                                                                                                                                                                                                                                                   & \textbf{Round:1, Level:2}                                                                                                                                                                                                                                                                                       & \textbf{Round:1, Level:3}                                                                                                                                                                                                                                                                                       & \textbf{Round:1, Level:4}                                                                                                                                                                                                                                                                                                                                                                                                                                                                                                                                                      \\
			\begin{tabular}[c]{@{}c@{}}$M_1 \xrightarrow{Type-1} M_2$\\ $M_3 \xrightarrow{Type-1} M_4$\\ $M_5 \xrightarrow{Type-1} M_6$\\ $M_7 \xrightarrow{Type-1} M_8$\\ $M_9 \xrightarrow{Type-1} M_{10}$\\ $M_{11} \xrightarrow{Type-1} M_{12}$\\ $M_{13} \xrightarrow{Type-1} M_{14}$\end{tabular} & \begin{tabular}[c]{@{}c@{}}$M_1 \xrightarrow{Type-1} M_4$\\ $M_3 \xrightarrow{Type-1} M_6$\\ $M_5 \xrightarrow{Type-1} M_8$\\ $M_7 \xrightarrow{Type-1} M_{10}$\\ $M_9 \xrightarrow{Type-1} M_{12}$\\ $M_{11} \xrightarrow{Type-1} M_{14}$\\ $M_{13} \xrightarrow{Type-1} M_2$\end{tabular}                     & \begin{tabular}[c]{@{}c@{}}$M_1 \xrightarrow{Type-1} M_6$\\ $M_3 \xrightarrow{Type-1} M_8$\\ $M_5 \xrightarrow{Type-1} M_{10}$\\ $M_7 \xrightarrow{Type-1} M_{12}$\\ $M_9 \xrightarrow{Type-1} M_{14}$\\ $M_{11} \xrightarrow{Type-1} M_2$\\ $ M_{13} \xrightarrow{Type-1} M_4$\end{tabular}                    & \begin{tabular}[c]{@{}c@{}}$M_1 \xrightarrow{Type-1} M_8$\\ $M_3 \xrightarrow{Type-1} M_{10}$\\ $M_5 \xrightarrow{Type-1} M_{12}$\\ $M_7 \xrightarrow{Type-1} M_{14}$\\ $M_9 \xrightarrow{Type-1} M_2$\\ $M_{11} \xrightarrow{Type-1} M_4$\\ $M_{13} \xrightarrow{Type-1} M_6$\end{tabular}               \end{tabular}
		\end{table}
			
		\begin{tabular}{ccc}
			\centering
			\textbf{Round:1, Level:5}                                                                                                                                                                                                                                                                                       & \textbf{Round:1, Level:6}                                                                                                                                                                                                                                                                                       & \textbf{Round:1, Level:7}  \\
			\multicolumn{1}{l}\centering{\begin{tabular}[c]{@{}l@{}}$M_1 \xrightarrow{Type-1} M_{12}$\\ $M_3 \xrightarrow{Type-1} M_{14}$\\ $M_5 \xrightarrow{Type-1} M_2$\\ $M_7 \xrightarrow{Type-1} M_4$\\ $M_9 \xrightarrow{Type-1} M_6$\\ $M_{11} \xrightarrow{Type-1} M_8$\\ $M_{13} \xrightarrow{Type-1} M_{10}$\end{tabular}} & \multicolumn{1}{l}\centering{\begin{tabular}[c]{@{}l@{}}$M_1 \xrightarrow{Type-1} M_{10}$\\ $M_3 \xrightarrow{Type-2} M_{12}$\\ $M_5 \xrightarrow{Type-1} M_{14}$\\ $M_7 \xrightarrow{Type-2} M_2$\\ $M_9 \xrightarrow{Type-1} M_4$\\ $M_{11} \xrightarrow{Type-2} M_6$\\ $M_{13} \xrightarrow{Type-1} M_8$\end{tabular}} & \multicolumn{1}{l}\centering{\begin{tabular}[c]{@{}l@{}}$M_1 \xrightarrow{Type-1} M_{14}$\\ $M_{12} \xrightarrow{Type-1} M_{10}$\\ $M_5 \xrightarrow{Type-1} M_7$\\ $M_2 \xrightarrow{Type-1} M_4$\\ $M_9 \xrightarrow{Type-1} M_{11}$\\ $M_6 \xrightarrow{Type-1} M_8$\\ $M_{13} \xrightarrow{Type-1} M_3$\end{tabular}} 
		\\  & &\\                   
			\textbf{Round:2, Level:1}                                                                                                                                                                                                                                                                   & \textbf{Round:2, Level:2}                                                                                                                                                                                                                                                                                       & \textbf{Round:2, Level:3}                                                                                                                                                                                                                                                                                                                                                                                                                                                                                                                                                                                                                                                                                                                                                                                                                          \\
			\begin{tabular}[c]{@{}c@{}}$M_1 \xrightarrow{Type-1} M_7$\\ $M_{12} \xrightarrow{Type-1} M_8$\\ $M_5 \xrightarrow{Type-1} M_4$\\ $M_2 \xrightarrow{Type-1} M_{10}$\\ $M_9 \xrightarrow{Type-1} M_3$\\ $M_6 \xrightarrow{Type-1} M_{14}$\\ $M_{13} \xrightarrow{Type-1} M_{11}$\end{tabular} & \begin{tabular}[c]{@{}c@{}}$M_1 \xrightarrow{Type-2} M_{11}$\\ $M_{12} \xrightarrow{Type-1} M_4$\\ $M_5 \xrightarrow{Type-2} M_3$\\ $M_2 \xrightarrow{Type-1} M_{14}$\\ $M_9 \xrightarrow{Type-1} M_8$\\ $M_6 \xrightarrow{Type-1} M_{10}$\\ $M_{13} \xrightarrow{Type-1} M_7$\end{tabular}                     & \begin{tabular}[c]{@{}c@{}}$M_{11} \xrightarrow{Type-2} M_{10}$\\ $M_{12} \xrightarrow{Type-1} M_{14}$\\ $M_3 \xrightarrow{Type-1} M_1$\\ $M_2 \xrightarrow{Type-2} M_8$\\ $M_9 \xrightarrow{Type-1} M_7$\\ $M_6 \xrightarrow{Type-1} M_4$\\ $M_{13} \xrightarrow{Type-2} M_5$\end{tabular}                   
		\\     & &\\                              \textbf{Round:3, Level:1}                                                                                                                                                                                                                                                                                       & \textbf{Round:3, Level:2}                                                                                                                                                                                                                                               & \textbf{Round:4, Level:1}                                                                                                                                                                                                                                              \\
			\begin{tabular}[c]{@{}c@{}}$M_{10} \xrightarrow{Type-1} M_4$\\ $M_{12} \xrightarrow{Type-1} M_{13}$\\ $M_3 \xrightarrow{Type-1} M_7$\\ $M_8 \xrightarrow{Type-1} M_{14}$\\ $M_9 \xrightarrow{Type-1} M_1$\\ $M_6 \xrightarrow{Type-1} M_2$\\ $M_5 \xrightarrow{Type-1} M_1$\end{tabular}                        & \begin{tabular}[c]{@{}l@{}}$M_{10} \xrightarrow{Type-2} M_{14}$\\ $M_{12} \xrightarrow{Type-1} M_2$\\ $M_3 \xrightarrow{Type-1} M_{11}$\\ $M_8 \xrightarrow{Type-1} M_4$\\ $M_9 \xrightarrow{Type-1} M_{13}$\\ $M_6 \xrightarrow{Type-2} M_7$\\ $M_5 \xrightarrow{Type-2} M_1$\end{tabular} & \begin{tabular}[c]{@{}l@{}}$M_{14} \xrightarrow{Type-3} M_4$\\ $M_{12} \xrightarrow{Type-3} M_6$\\ $M_3 \xrightarrow{Type-3} M_2$\\ $M_8 \xrightarrow{Type-3} M_{10}$\\ $M_9 \xrightarrow{Type-3} M_5$\\ $M_7 \xrightarrow{Type-3} M_{11}$\\ $M_1 \xrightarrow{Type-3} M_{13}$\end{tabular}\\
			& & \\
			\multicolumn{3}{l}{\textbf{Number of \textit{Flips}$=11=\lceil F_{28} \rceil$.}}	
	\end{tabular}

	\section{Tabular IPL Schedule}
	\label{IPL Schedule}
	
	In this section, we present a schedule of \textbf{Indian Premier League (IPL)} using our algorithm\ref{alg:TTP-2}. IPL is a Double Round-robin Tournament of eight teams. The proposed schedule is presented in Table-1 where the teams are represented as the following:\\	
	
	\begin{tabular}{ll}
		\textbf{KOL} $\rightarrow$ Kolkata Knight Riders & \textbf{MUM} $\rightarrow$ Mumbai Indians\\ \textbf{CHE} $\rightarrow$ Chennai Super Kings & \textbf{BANG} $\rightarrow$ Royal Challengers Bangalore\\ \textbf{RAJ} $\rightarrow$ Rajasthan Royals & \textbf{DEL} $\rightarrow$ Delhi Capitals\\ \textbf{HYD} $\rightarrow$ Sunrisers Hyderabad & \textbf{PUN} $\rightarrow$ Kings XI Punjab  
	\end{tabular}
		
	\begin{table}[h]
		\centering
		\caption{\textbf{\textit{Proposed Indian Premier League Schedule}}}
		\centering
		\begin{tabular}{ccc}
			\begin{tabular}{|c|c|}
				\hline
				\textbf{Match Day} & \textbf{1}\\
				\hline
				\hline
				\textit{Away} & \textit{Home} \\
				\hline 
				\hline
				MUM & KOL \\ 
				\hline 
				HYD & RAJ \\ 
				\hline 
				CHE & DEL \\ 
				\hline 
				BANG & PUN \\ 
				\hline
				\hline	
				\hline
				\textbf{Match Day} & \textbf{4}\\
				\hline
				\hline
				\textit{Away} & \textit{Home}\\
				\hline
				\hline 
				RAJ & MUM \\ 
				\hline 
				KOL & HYD \\ 
				\hline
				DEL & BANG \\ 
				\hline 
				PUN & CHE \\ 
				\hline
				\hline	
				\hline
				\textbf{Match Day} & \textbf{7}\\
				\hline
				\hline
				\textit{Away} & \textit{Home}\\
				\hline
				\hline  
				PUN & HYD \\ 
				\hline 
				DEL & MUM \\ 
				\hline 
				RAJ & BANG \\ 
				\hline
				KOL & CHE \\ 
				\hline
				\hline	
				\hline
				\textbf{Match Day} & \textbf{10}\\
				\hline
				\hline
				\textit{Away} & \textit{Home}\\
				\hline
				\hline  
				MUM & HYD \\ 
				\hline 
				KOL & RAJ \\ 
				\hline 
				BAG & CHE \\ 
				\hline 
				PUN & DEL \\
				\hline
				\hline	
				\hline
				\textbf{Match Day} & \textbf{13}\\
				\hline
				\hline
				\textit{Away} & \textit{Home}\\
				\hline
				\hline   
				MUM & BANG \\ 
				\hline 
				HYD & CHE \\ 
				\hline 
				KOL & PUN \\ 
				\hline 
				RAJ & DEL \\ 
				\hline
			\end{tabular}
			& \!\!\!\!\!
			\begin{tabular}{|c|c|}	
				\hline
				\textbf{Match Day} & \textbf{2}\\
				\hline
				\hline
				\textit{Away} & \textit{Home} \\
				\hline 
				\hline
				MUM & RAJ \\ 
				\hline 
				HYD & KOL \\ 
				\hline 
				CHE & PUN \\ 
				\hline 
				BANG & DEL \\ 
				\hline
				\hline			
				\hline
				\textbf{Match Day} & \textbf{5}\\
				\hline
				\hline
				\textit{Away} & \textit{Home} \\
				\hline
				\hline
				MUM & DEL \\ 
				\hline 
				HYD & PUN \\ 
				\hline 
				CHE & KOL \\ 
				\hline 
				BANG & RAJ \\ 
				\hline 
				\hline			
				\hline
				\textbf{Match Day} & \textbf{8}\\
				\hline
				\hline
				\textit{Away} & \textit{Home} \\
				\hline
				\hline
				DEL & HYD \\ 
				\hline 
				PUN & MUM \\ 
				\hline 
				CHE & RAJ \\ 
				\hline 
				BANG & KOL \\ 
				\hline  
				\hline			
				\hline
				\textbf{Match Day} & \textbf{11}\\
				\hline
				\hline
				\textit{Away} & \textit{Home}\\
				\hline
				\hline  
				BANG & MUM \\ 
				\hline 
				CHE & HYD \\ 
				\hline
				PUN & KOL \\ 
				\hline 
				DEL & RAJ \\ 
				\hline  
				\hline			
				\hline
				\textbf{Match Day} & \textbf{14}\\
				\hline
				\hline
				\textit{Away} & \textit{Home}\\
				\hline
				\hline   
				CHE & MUM \\ 
				\hline 
				BANG & HYD \\ 
				\hline 
				DEL & KOL \\ 
				\hline 
				PUN & RAJ \\ 
				\hline   
			\end{tabular}
			& \!\!\!\!\!
			\begin{tabular}{cc}
				\hline
				\multicolumn{1}{|c|}{\textbf{Match Day}} & \multicolumn{1}{|c|}{\textbf{3}}\\
				\hline
				\hline
				\multicolumn{1}{|c|}{\textit{Away}} & \multicolumn{1}{|c|}{\textit{Home}} \\
				\hline 
				\hline 
				\multicolumn{1}{|c|}{RAJ} & \multicolumn{1}{|c|}{HYD} \\ 
				\hline 
				\multicolumn{1}{|c|}{KOL} & \multicolumn{1}{|c|}{MUM} \\ 
				\hline 
				\multicolumn{1}{|c|}{DEL} & \multicolumn{1}{|c|}{CHE} \\ 
				\hline 
				\multicolumn{1}{|c|}{PUN} & \multicolumn{1}{|c|}{BANG} \\ 
				\hline
				\hline
				\hline
				\multicolumn{1}{|c|}{\textbf{Match Day}} & \multicolumn{1}{|c|}{\textbf{6}}\\
				\hline
				\hline
				\multicolumn{1}{|c|}{\textit{Away}} & \multicolumn{1}{|c|}{\textit{Home}}\\
				\hline
				\hline 
				\multicolumn{1}{|c|}{MUM} & \multicolumn{1}{|c|}{PUN} \\ 
				\hline 
				\multicolumn{1}{|c|}{HYD} & \multicolumn{1}{|c|}{DEL} \\ 
				\hline 
				\multicolumn{1}{|c|}{RAJ} & \multicolumn{1}{|c|}{CHE} \\ 
				\hline 
				\multicolumn{1}{|c|}{KOL} & \multicolumn{1}{|c|}{BANG} \\ 
				\hline
				\hline
				\hline
				\multicolumn{1}{|c|}{\textbf{Match Day}} & \multicolumn{1}{|c|}{\textbf{9}}\\
				\hline
				\hline
				\multicolumn{1}{|c|}{\textit{Away}} & \multicolumn{1}{|c|}{\textit{Home}}\\
				\hline
				\hline  
				\multicolumn{1}{|c|}{MUM} & \multicolumn{1}{|c|}{CHE} \\ 
				\hline 
				\multicolumn{1}{|c|}{HYD} & \multicolumn{1}{|c|}{BANG} \\ 
				\hline 
				\multicolumn{1}{|c|}{KOL} & \multicolumn{1}{|c|}{DEL} \\ 
				\hline 
				\multicolumn{1}{|c|}{RAJ} & \multicolumn{1}{|c|}{PUN} \\ 
				\hline
				\hline
				\hline
				\multicolumn{1}{|c|}{\textbf{Match Day}} & \multicolumn{1}{|c|}{\textbf{12}}\\
				\hline
				\hline
				\multicolumn{1}{|c|}{\textit{Away}} & \multicolumn{1}{|c|}{\textit{Home}} \\
				\hline
				\hline		 
				\multicolumn{1}{|c|}{HYD} & \multicolumn{1}{|c|}{MUM} \\ 
				\hline 
				\multicolumn{1}{|c|}{RAJ} & \multicolumn{1}{|c|}{KOL} \\ 
				\hline 
				\multicolumn{1}{|c|}{CHE} & \multicolumn{1}{|c|}{BANG} \\ 
				\hline 
				\multicolumn{1}{|c|}{DEL} & \multicolumn{1}{|c|}{PUN} \\  
				\hline
				\hline
				&\\
				& \\
				
				&  \\
				
				& \\ 
				
				& \\ 
				
				& \\ 
				
				& \\  
				
			\end{tabular}\\
		& &\\
		& &\\
		\multicolumn{3}{l}{This schedule gives $15\%$ better result than actual \textbf{IPL-2019} schedule} \\
		\multicolumn{3}{l}{in terms of total distance traveled by all the teams.} 
		\end{tabular}
	\end{table}

\end{document}